\documentclass{article} 
\usepackage{nips10submit_e,times}

\usepackage{amsmath}%
\usepackage{amsfonts}%
\usepackage{amssymb}%
\usepackage{graphicx}
\usepackage{url}

\newtheorem{corollary}{Corollary}

\newtheorem{lemma}{Lemma}

\newtheorem{proposition}{Proposition}

\newenvironment{proof}[1][Proof]{\noindent\textbf{#1.} }{\ \rule{0.5em}{0.5em}}

\title{Improving the Asymptotic Performance of Markov Chain Monte-Carlo by Inserting Vortices}

\author{
Yi Sun \\
IDSIA\\
Galleria 2, Manno CH-6928, Switzerland\\
\texttt{yi@idsia.ch} \\
\And
Faustino Gomez \\
IDSIA \\
Galleria 2, Manno CH-6928, Switzerland\\
\texttt{tino@idsia.ch} \\
\And
J\"urgen Schmidhuber \\
IDSIA \\
Galleria 2, Manno CH-6928, Switzerland\\
\texttt{juergen@idsia.ch} \\
}

\nipsfinalcopy 

\begin{document}

\maketitle

\begin{abstract}


We present a new way of converting a reversible finite Markov chain into a
non-reversible one, with a theoretical guarantee that the asymptotic variance
of the MCMC estimator based on the non-reversible chain is reduced.
The method is applicable to any reversible chain whose states are not connected
through a tree, and can be interpreted graphically as inserting vortices into
the state transition graph.
Our result confirms that non-reversible chains are fundamentally better
than reversible ones in terms of asymptotic performance, and suggests interesting
directions for further improving MCMC.

\end{abstract}

\section{Introduction}

Markov Chain Monte Carlo (MCMC) methods have gained enormous popularity over a
wide variety of research fields \cite{Neal 96,Andrieu 03}, owing to their
ability to compute expectations with respect to complex, high dimensional
probability distributions.
An MCMC estimator can be based on any ergodic Markov chain with the
distribution of interest as its stationary distribution. However, the choice
of Markov chain greatly affects the performance of the estimator, in
particular the accuracy achieved with a pre-specified number of samples
\cite{Neal 04}.

In general, the efficiency of an MCMC estimator is determined by two factors:
i) how fast the chain converges to its stationary distribution, i.e., the
mixing rate \cite{Szakdolgozat 06}, and ii) once the chain reaches its
stationary distribution, how much the estimates fluctuate based on
trajectories of finite length, which is characterized by the asymptotic
variance.
In this paper, we consider the latter criteria. Previous theory concerned with
reducing asymptotic variance has followed two main tracks. The first focuses
on reversible chains, and is mostly based on the theorems of Peskun
\cite{Peskun 73} and Tierney \cite{Tierney 98}, which state that if a
reversible Markov chain is modified so that the probability of staying in the
same state is reduced, then the asymptotic variance can be decreased. A number
of methods have been proposed, particularly in the context of
Metropolis-Hastings method, to encourage the Markov chain to move away from
the current state, or its adjacency in the continuous case \cite{Duane 87,Liu
96}. The second track, which was explored just recently, studies
non-reversible chains. Neal proved in \cite{Neal 04} that starting from any
finite-state reversible chain, the asymptotic variance of a related
non-reversible chain, with reduced probability of back-tracking to the
immediately previous state, will not increase, and typically decrease. Several
methods have been proposed by Murray based on this idea \cite{Murray 07}.

Neal's result suggests that non-reversible chains may be fundamentally better
than reversible ones in terms of the asymptotic performance. In this paper, we
follow up this idea by proposing a new way of converting reversible chains
into non-reversible ones which, unlike in Neal's method, are defined on the
state space of the reversible chain, with the theoretical guarantee that the
asymptotic variance of the associated MCMC estimator is reduced. Our method is
applicable to any non-reversible chain whose state transition graph contains
loops, including those whose probability of staying in the same state is zero
and thus cannot be improved using Peskun's theorem. The method also admits an
interesting graphical interpretation which amounts to inserting `vortices'
into the state transition graph of the original chain. Our result suggests a
new and interesting direction for improving the asymptotic performance of MCMC.

The rest of the paper is organized as follows: section 2 reviews some
background concepts and results; section 3 presents the main theoretical
results, together with the graphical interpretation; section 4 provides
a simple yet illustrative example and explains the intuition behind the
results; section 5 concludes the paper.

\section{Preliminaries}

Suppose we wish to estimate the expectation of some real valued function $f$
over domain $\mathcal{S}$, with respect to a probability distribution $\pi$,
whose value may only be known to a multiplicative constant. Let $A$ be a
transition operator of an ergodic\footnote{Strictly speaking, the ergodic
assumption is not necessary for the MCMC\ estimator to work, see \cite{Neal
04}. However, we make the assumption to simplify the analysis.} Markov chain
with stationary distribution $\pi$, i.e.,%
\begin{equation}
\pi\left(  x\right)  A\left(  x\rightarrow y\right)  =\pi\left(  y\right)
B\left(  y\rightarrow x\right)  \text{, }\forall x,y\in\mathcal{S}\text{,}
\label{eq-10}%
\end{equation}
where $B$ is the reverse operator as defined in~\cite{Murray 07}. The
expectation can then be estimated through the MCMC estimator%
\begin{equation}
\mu_{T}=\frac{1}{T}\sum\nolimits_{t=1}^{T}f\left(  x_{t}\right)  \text{,}
\label{eq-20}%
\end{equation}
where $x_{1},\cdots,x_{T}$ is a trajectory sampled from the Markov chain. The
asymptotic variance of $\mu_{T}$, with respect to transition operator $A$ and
function $f$ is defined as%
\begin{equation}
\sigma_{A}^{2}\left(  f\right)  =\lim_{T\rightarrow\infty}T\mathbb{V}\left[
\mu_{T}\right]  \text{,} \label{eq-30}%
\end{equation}
where $\mathbb{V}\left[  \mu_{T}\right]  $ denotes the variance of $\mu_{T}$.
Since the chain is ergodic, $\sigma_{A}^{2}\left(  f\right)  $ is well-defined
following the central limit theorem, and does not depend on the distribution
of the initial point. Roughly speaking, asymptotic variance has the meaning
that the mean square error of the estimates based on $T$ consecutive states of
the chain would be approximately $\frac{1}{T}\sigma_{A}^{2}\left(  f\right)
$, after a sufficiently long period of "burn in" such that the chain is close
enough to its stationary distribution. Asymptotic variance can be used to
compare the asymptotic performance of MCMC estimators based on different
chains with the same stationary distribution, where smaller asymptotic
variance indicates that, asymptotically, the MCMC estimator requires fewer
samples to reach a specified accuracy.

Under the ergodic assumption, the asymptotic variance can be written as%
\begin{equation}
\sigma_{A}^{2}\left(  f\right)  =\mathbb{V}\left[  f\right]  +\sum
\nolimits_{\tau=1}^{\infty}\left(  c_{A,f}\left(  \tau\right)  +c_{B,f}\left(
\tau\right)  \right)  \text{,} \label{eq-40}%
\end{equation}
where
\[
c_{A,f}\left(  \tau\right)  =\mathbb{E}_{A}\left[  f\left(  x_{t}\right)
f\left(  x_{t+\tau}\right)  \right]  -\mathbb{E}_{A}\left[  f\left(
x_{t}\right)  \right]  \mathbb{E}\left[  f\left(  x_{t+\tau}\right)  \right]
\]
is the covariance of the function value between two states that are $\tau$
time steps apart in the trajectory of the Markov chain with transition
operator $A$. Note that $\sigma_{A}^{2}\left(  f\right)  $ depends on both $A$
and its reverse operator $B$, and $\sigma_{A}^{2}\left(  f\right)  =\sigma
_{B}^{2}\left(  f\right)  $ since $A$ is also the reverse operator of $B$ by definition.

In this paper, we consider only the case where $\mathcal{S}$ is finite, i.e.,
$\mathcal{S}=\left\{  1,\cdots,S\right\}  $, so that the transition operators
$A$ and $B$, the stationary distribution $\pi$, and the function $f$ can all
be written in matrix form. Let $\pi=\left[  \pi\left(  1\right)  ,\cdots
,\pi\left(  S\right)  \right]  ^{\top}$, $f=\left[  f\left(  1\right)
,\cdots,f\left(  S\right)  \right]  ^{\top}$, $A_{i,j}=A\left(  i\rightarrow
j\right)  $, $B_{i,j}=B\left(  i\rightarrow j\right)  $. The asymptotic
variance can thus be written as%
\[
\sigma_{A}^{2}\left(  f\right)  =\mathbb{V}\left[  f\right]  +\sum
\nolimits_{\tau=1}^{\infty}f^{\top}\left(  QA^{\tau}+QB^{\tau}-2\pi\pi^{\top
}\right)  f\text{,}%
\]
with $Q=\operatorname*{diag}\left\{  \pi\right\}  $. Since $B$ is the reverse
operator of $A$, $QA=B^{\top}Q$. Also, from the ergodic assumption,%
\[
\lim_{\tau\rightarrow\infty}A^{\tau}=\lim_{\tau\rightarrow\infty}B^{\tau
}=R\text{,}%
\]
where $R=\mathbf{1}\pi^{\top}$ is a square matrix in which every row is
$\pi^{\top}$. It follows that the asymptotic variance can be represented by
Kenney's formula \cite{Kenney 63} in the non-reversible case:%
\begin{equation}
\sigma_{A}^{2}\left(  f\right)  =\mathbb{V}\left[  f\right]  +2\left(
Qf\right)  ^{\top}\left[  \Lambda^{-}\right]  _{H}\left(  Qf\right)
-2f^{\top}Qf\text{,} \label{eq-50}%
\end{equation}
where $\left[  \cdot\right]  _{H}$ denotes the Hermitian (symmetric) part of a
matrix, and $\Lambda=Q+\pi\pi^{\top}-J$, with $J=QA$ being the joint
distribution of two consecutive states.

\section{Improving the asymptotic variance}

It is clear from Eq.\ref{eq-50} that the transition operator $A$ affects the
asymptotic variance only through term $\left[  \Lambda^{-}\right]  _{H}$. If
the chain is reversible, then $J$ is symmetric, so that $\Lambda$ is also
symmetric, and therefore comparing the asymptotic variance of two
MCMC\ estimators becomes a matter of comparing their $J$, namely,
if\footnote{For symmetric matrices $X$ and $Y$, we write $X\preceq Y$ if $Y-X$
is positive semi-definite, and $X\prec Y$ if $Y-X$ is positive definite.}
$J\preceq J^{\prime}=QA^{\prime}$, then $\sigma_{A}^{2}\left(  f\right)
\leq\sigma_{A^{\prime}}^{2}\left(  f\right)  $, for any $f$. This leads to a
simple proof of Peskun's theorem in the discrete case \cite{Li 05}.

In the case where the Markov chain is non-reversible, i.e., $J$ is asymmetric,
the analysis becomes much more complicated. We start by providing a sufficient
and necessary condition in section \ref{general}, which transforms the
comparison of asymptotic variance based on arbitrary finite Markov chains into
a matrix ordering problem, using a result from matrix analysis. In section
\ref{special-case}, a special case is identified, in which the asymptotic
variance of a reversible chain is compared to that of a non-reversible one
whose joint distribution over consecutive states is that of the reversible
chain plus a skew-Hermitian matrix. We prove that the resulting non-reversible
chain has smaller asymptotic variance, and provide a necessary and sufficient
condition for the existence of such non-zero skew-Hermitian matrices. Finally
in section \ref{graphical}, we provide a graphical interpretation of the result.

\subsection{The general case\label{general}}

From Eq.\ref{eq-50} we know that comparing the asymptotic variances of two
MCMC estimators is equivalent to comparing their $\left[  \Lambda^{-}\right]
_{H}$. The following result from \cite{Wen 05,Mathias 92} allows us to write
$\left[  \Lambda^{-}\right]  _{H}$ in terms of the symmetric and asymmetric
parts of $\Lambda$.

\begin{lemma}
\label{lemma-1}If a matrix $X$ is invertible, then $\left[  X^{-}\right]
_{H}^{-}=\left[  X\right]  _{H}+\left[  X\right]  _{S}^{\top}\left[  X\right]
_{H}^{-}\left[  X\right]  _{S}$, where $\left[  X\right]  _{S}$ is the skew
Hermitian part of $X$.
\end{lemma}

From Lemma \ref{lemma-1}, it follows immediately that in the discrete case,
the comparison of MCMC estimators based on two Markov chains with the same
stationary distribution can be cast as a different problem of matrix
comparison, as stated in the following proposition.

\begin{proposition}
\label{prop-1}Let $A$, $A^{\prime}$ be two transition operators of ergodic
Markov chains with stationary distribution $\pi$. Let $J=QA$, $J^{\prime
}=QA^{\prime}$, $\Lambda=Q+\pi\pi^{\top}-J$, $\Lambda^{\prime}=Q+\pi\pi^{\top
}-J^{\prime}$. Then the following three conditions are equivalent:

\begin{description}
\item[1)] $\sigma_{A}^{2}\left(  f\right)  \leq\sigma_{A^{\prime}}^{2}\left(
f\right)  $ for any $f$

\item[2)] $\left[  \Lambda^{-}\right]  _{H}\preceq\left[  \left(
\Lambda^{\prime}\right)  ^{-}\right]  _{H}$

\item[3)] $\left[  J\right]  _{H}-\left[  J\right]  _{S}^{\top}\left[
\Lambda\right]  _{H}^{-}\left[  J\right]  _{S}\preceq\left[  J^{\prime
}\right]  _{H}-\left[  J^{\prime}\right]  _{S}^{\top}\left[  \Lambda^{\prime
}\right]  _{H}^{-}\left[  J^{\prime}\right]  _{S}$
\end{description}
\end{proposition}

\begin{proof}
First we show that $\Lambda$ is invertible. Following the steps in \cite{Li
05}, for any $f\neq0$,%
\begin{align*}
f^{\top}\Lambda f  &  =f^{\top}\left[  \Lambda\right]  _{H}f=f^{\top}\left(
Q+\pi\pi^{\top}-J\right)  f\\
&  =\frac{1}{2}\mathbb{E}\left[  \left(  f\left(  x_{t}\right)  -f\left(
x_{t+1}\right)  \right)  ^{2}\right]  +\mathbb{E}\left[  f\left(
x_{t}\right)  \right]  ^{2}>0\text{,}%
\end{align*}
thus $\left[  \Lambda\right]  _{H}\succ0$, and $\Lambda$ is invertible since
$\Lambda f\neq0$ for any $f\neq0$.

Condition 1) and 2) are equivalent by definition. We now prove 2) is
equivalent to 3). By Lemma \ref{lemma-1},%
\[
\left[  \Lambda^{-}\right]  _{H}\preceq\left[  \left(  \Lambda^{\prime
}\right)  ^{-}\right]  _{H}\Longleftrightarrow\left[  \Lambda\right]
_{H}+\left[  \Lambda\right]  _{S}^{\top}\left[  \Lambda\right]  _{H}\left[
\Lambda\right]  _{S}\succeq\left[  \Lambda^{\prime}\right]  _{H}+\left[
\Lambda^{\prime}\right]  _{S}^{\top}\left[  \Lambda^{\prime}\right]
_{H}\left[  \Lambda^{\prime}\right]  _{S}\text{,}%
\]
the result follows by noticing that $\left[  \Lambda\right]  _{H}=Q+\pi
\pi^{\top}-\left[  J\right]  _{H}$ and $\left[  \Lambda\right]  _{S}=-\left[
J\right]  _{S}$.
\end{proof}

\subsection{A special case\label{special-case}}

Generally speaking, the conditions in Proposition \ref{prop-1} are very hard
to verify, particularly because of the term $\left[  J\right]  _{S}^{\top
}\left[  \Lambda\right]  _{H}^{-}\left[  J\right]  _{S}$. Here we focus on a
special case where $\left[  J^{\prime}\right]  _{S}=0$, and $\left[
J^{\prime}\right]  _{H}=J^{\prime}=\left[  J\right]  _{H}$. This amounts to
the case where the second chain is reversible, and its transition operator is
the average of the transition operator of the first chain and the associated
reverse operator. The result is formalized in the following corollary.

\begin{corollary}
\label{cor-1}Let $T$ be a reversible transition operator of a Markov chain
with stationary distribution $\pi$. Assume there is some $H$ that satisfies

\begin{description}
\item[Condition I.] $\mathbf{1}^{\top}H=0$, $H\mathbf{1}=0$, $H=-H^{\top}$,
and\footnote{We write $\mathbf{1}$ for the $S$-dimensional column vector of
$1$'s.}

\item[Condition II.] $T\pm Q^{-}H$ are valid transition matrices.
\end{description}

Denote $A=T+Q^{-}H$, $B=T-Q^{-}H$, then

\begin{description}
\item[1)] $A$ preserves $\pi$, and $B$ is the reverse operator of $A$.

\item[2)] $\sigma_{A}^{2}\left(  f\right)  =\sigma_{B}^{2}\left(  f\right)
\leq\sigma_{T}^{2}\left(  f\right)  $ for any $f$.

\item[3)] If $H\neq0$, then there is some $f$, such that $\sigma_{A}%
^{2}\left(  f\right)  <\sigma_{T}^{2}\left(  f\right)  $.

\item[4)] If $A_{\varepsilon}=T+\left(  1+\varepsilon\right)  Q^{-}H$ is valid
transition matrix, $\varepsilon>0$, then $\sigma_{A_{\varepsilon}}^{2}\left(
f\right)  \leq\sigma_{A}^{2}\left(  f\right)  $.
\end{description}
\end{corollary}

\begin{proof}
For 1), notice that $\pi^{\top}T=\pi^{\top}$, so%
\[
\pi^{\top}A=\pi^{\top}T+\pi^{\top}Q^{-}H=\pi^{\top}+\mathbf{1}^{\top}%
H=\pi^{\top}\text{,}%
\]
and similarly for $B$. Moreover%
\[
QA=QT+H=\left(  QT-H\right)  ^{\top}=\left(  Q\left(  T-Q^{-}H\right)
\right)  ^{\top}=\left(  QB\right)  ^{\top}\text{,}%
\]
thus $B$ is the reverse operator of $A$.

For 2), $\sigma_{A}^{2}\left(  f\right)  =\sigma_{B}^{2}\left(  f\right)  $
follows from Eq.\ref{eq-50}. Let $J^{\prime}=QT$, $J=QA$. Note that $\left[
J\right]  _{S}=H$,%
\[
J^{\prime}=QT=\frac{1}{2}\left(  QA+QB\right)  =\left[  QA\right]
_{H}=\left[  J\right]  _{H}\text{,}%
\]
and $\left[  \Lambda\right]  _{H}\succ0$ thus $H^{\top}\left[  \Lambda\right]
_{H}^{-}H\succeq0$ from Proposition \ref{prop-1}. It follows that $\sigma
_{A}^{2}\left(  f\right)  \leq\sigma_{T}^{2}\left(  f\right)  $ for any $f$.

For 3), write $X=\left[  \Lambda\right]  _{H}$,%
\[
\left[  \Lambda^{-}\right]  _{H}=\left(  X+H^{\top}X^{-}H\right)  ^{-}%
=X^{-}-X^{-}H^{\top}\left(  X+HX^{-}H^{\top}\right)  ^{-}HX^{-}\text{.}%
\]
Since $X\succ0$, $HX^{-}H^{\top}\succeq0$, one can write $\left(
X+HX^{-}H^{\top}\right)  ^{-}=\sum_{s=1}^{S}\lambda_{s}e_{s}e_{s}^{\top}$,
with $\lambda_{s}>0$, $\forall s$. Thus%
\[
H^{\top}\left(  X+HX^{-}H^{\top}\right)  ^{-}H=\sum\nolimits_{s=1}^{S}%
\lambda_{s}He_{s}\left(  He_{s}\right)  ^{\top}\text{.}%
\]
Since $H\neq0$, there is at least one $s^{\ast}$, such that $He_{s^{\ast}}%
\neq0$. Let $f=Q^{-}XHe_{s^{\ast}}$, then%
\begin{align*}
\frac{1}{2}\left[  \sigma_{T}^{2}\left(  f\right)  -\sigma_{A}^{2}\left(
f\right)  \right]   &  =\left(  Qf\right)  ^{\top}\left[  X^{-}-\left(
X+H^{\top}X^{-}H\right)  ^{-}\right]  \left(  Qf\right) \\
&  =\left(  Qf\right)  ^{\top}X^{-}H^{\top}\left(  X+HX^{-}H^{\top}\right)
^{-}HX^{-}\left(  Qf\right) \\
&  =\left(  He_{s^{\ast}}\right)  ^{\top}\sum\nolimits_{s=1}^{S}\lambda
_{s}He_{s}\left(  He_{s}\right)  ^{\top}\left(  He_{s^{\ast}}\right) \\
&  =\lambda_{s}\left\Vert He_{s^{\ast}}\right\Vert ^{4}+\sum\nolimits_{s\neq
s^{\ast}}\lambda_{s}\left(  e_{s^{\ast}}^{\top}H^{\top}He_{s}\right)
^{2}>0\text{.}%
\end{align*}

For 4), let $\Lambda_{\varepsilon}=Q+\pi\pi^{\top}-QA_{\varepsilon}$, then for
$\varepsilon>0$,
\[
\left[  \Lambda_{\varepsilon}^{-}\right]  _{H}=\left(  X+\left(
1+\varepsilon\right)  ^{2}H^{\top}X^{-}H\right)  ^{-}\preceq\left(  X+H^{\top
}X^{-}H\right)  ^{-}=\left[  \Lambda^{-}\right]  _{H}\text{,}%
\]
by Eq.\ref{eq-50}, we have $\sigma_{A_{\varepsilon}}^{2}\left(  f\right)
\leq\sigma_{A}^{2}\left(  f\right)  $ for any $f$.
\end{proof}

Corollary \ref{cor-1} shows that starting from a reversible Markov chain, as
long as one can find a non-zero $H$ satisfying Conditions I and II, then the
asymptotic performance of the MCMC estimator is guaranteed to improve. The
next question to ask is whether such an $H$ exists, and, if so, how to find
one. We answer this question by first looking at Condition I. The following
proposition shows that any $H$ satisfying this condition can be constructed systematically.

\begin{proposition}
\label{prop-2}Let $H$ be an $S$-by-$S$ matrix. $H$ satisfies Condition I if
and only if $H$ can be written as the linear combination of $\frac{1}%
{2}\left(  S-1\right)  \left(  S-2\right)  $ matrices, with each matrix of the
form%
\[
U_{i,j}=u_{i}u_{j}^{\top}-u_{j}u_{i}^{\top}\text{, }1\leq i<j\leq S-1\text{.}%
\]
Here $u_{1},\cdots,u_{S-1}$ are $S-1$ non-zero linearly independent vectors
satisfying $u_{s}^{\top}\mathbf{1}=0$.
\end{proposition}

\begin{proof}
Sufficiency. It is straightforward to verify that each $U_{i,j}$ is
skew-Hermitian and satisfies $U_{i,j}\mathbf{1}=0$. Such properties are
inherited by any linear combination of $U_{i,j}$.

Necessity. We show that there are at most $\frac{1}{2}\left(  S-1\right)
\left(  S-2\right)  $ linearly independent bases for all $H$ such that
$H=-H^{\top}$ and $H\mathbf{1}=0$. On one hand, any $S$-by-$S$ skew-Hermitian
matrix can be written as the linear combination of $\frac{1}{2}S\left(
S-1\right)  $ matrices of the form%
\[
V_{i,j}:\left\{  V_{i,j}\right\}  _{m,n}=\delta\left(  m,i\right)
\delta\left(  n,j\right)  -\delta\left(  n,i\right)  \delta\left(  m,j\right)
\text{,}%
\]
where $\delta$ is the standard delta function such that $\delta\left(
i,j\right)  =1$ if $i=j$ and $0$ otherwise. However, the constraint
$H\mathbf{1}=0$ imposes $S-1$ linearly independent constraints, which means
that out of $\frac{1}{2}S\left(  S-1\right)  $ parameters, only%
\[
\frac{1}{2}S\left(  S-1\right)  -\left(  S-1\right)  =\frac{1}{2}\left(
S-1\right)  \left(  S-2\right)
\]
are independent.

On the other hand, selecting two non-identical vectors from $u_{1}%
,\cdots,u_{S-1}$ results in $\dbinom{S-1}{2}=\frac{1}{2}\left(  S-1\right)
\left(  S-2\right)  $ different $U_{i,j}$. It has still to be shown that these
$U_{i,j}$ are linearly independent.

Assume%
\[
0=\sum_{1\leq i<j\leq S-1}%
\!\!\!\!\!%
\kappa_{i,j}U_{i,j}=\sum_{1\leq i<j\leq S-1}%
\!\!\!\!\!%
\kappa_{i,j}\left(  u_{i}u_{j}^{\top}-u_{j}u_{i}^{\top}\right)  \text{,
}\forall\kappa_{i,j}\in%
\mathbb{R}
\text{.}%
\]
Consider two cases: Firstly, assume $u_{1},\cdots,u_{S-1}$ are orthogonal,
i.e., $u_{i}^{\top}u_{j}=0$ for $i\neq j$. For a particular $u_{s}$,%
\begin{align*}
0  &  =\sum_{1\leq i<j\leq S-1}%
\!\!\!\!\!%
\kappa_{i,j}U_{i,j}u_{s}=\sum_{1\leq i<j\leq S-1}%
\!\!\!\!\!%
\kappa_{i,j}\left(  u_{i}u_{j}^{\top}-u_{j}u_{i}^{\top}\right)  u_{s}\\
&  =\sum_{1\leq i<s}\kappa_{i,s}u_{i}\left\Vert u_{s}^{\top}u_{s}\right\Vert
+\sum_{s<j\leq S-1}%
\!\!\!%
\kappa_{s,j}u_{j}\left\Vert u_{s}^{\top}u_{s}\right\Vert \text{.}%
\end{align*}
Since $\left\Vert u_{s}^{\top}u_{s}\right\Vert \neq0$, it follows that
$\kappa_{i,s}=\kappa_{s,j}=0$, for all $1\leq i<s<j\leq S-1$. This holds for
any $u_{s}$, so all $\kappa_{i,j}$ must be $0$, and therefore $U_{i,j}$ are
linearly independent by definition. Secondly, if $u_{1},\cdots,u_{S-1}$ are
not orthogonal, one can construct a new set of orthogonal vectors $\tilde
{u}_{1},\cdots,\tilde{u}_{S-1}$ from $u_{1},\cdots,u_{S-1}$ through
Gram--Schmidt orthogonalization, and create a different set of bases
$\tilde{U}_{i,j}$. It is easy to verify that each $\tilde{U}_{i,j}$ is a
linear combination of $U_{i,j}$. Since all $\tilde{U}_{i,j}$ are linearly
independent, it follows that $U_{i,j}$ must also be linearly independent.
\end{proof}

Proposition \ref{prop-2} confirms the existence of non-zero $H$ satisfying
Condition I. We now move to Condition II, which requires that both $QT+H$ and
$QT-H$ remain valid joint distribution matrices, i.e. all entries must be
non-negative and sum up to $1$. Since $\mathbf{1}^{\top}\left(  QT+H\right)
\mathbf{1}=1$ by Condition I, only the non-negative constraint needs to be considered.

It turns out that not all reversible Markov chains admit a non-zero $H$
satisfying both Condition I and II. For example, consider a Markov chain with
only two states. It is impossible to find a non-zero skew-Hermitian $H$ such
that $H\mathbf{1}=0$, because all 2-by-2 skew-Hermitian matrices are
proportional to ${\scriptstyle\left[
\begin{array}
[c]{rr}%
{0} & {-1}\\
{1} & {0}%
\end{array}
\right]  }$.

The next proposition gives the sufficient and necessary condition for the
existence of a non-zero $H$ satisfying both I and II. In particular, it shows
an interesting link between the existence of such $H$ and the connectivity of
the states in the reversible chain.

\begin{proposition}
\label{prop-3}Assume a reversible ergodic Markov chain with transition matrix
$T$ and let $J=QT$. The \emph{state transition graph} $\mathcal{G}_{T}$ is
defined as the undirected graph with node set $\mathcal{S}=\left\{
1,\cdots,S\right\}  $ and edge set $\left\{  \left(  i,j\right)
:J_{i,j}>0,1\leq i<j\leq S\right\}  $. Then there exists some non-zero $H$
satisfying Condition I and II, if and only if there is a loop in
$\mathcal{G}_{T}$.
\end{proposition}

\begin{proof}
Sufficiency: Without loss of generality, assume the loop is made of states
$1,2,\cdots,N$ and edges $\left(  1,2\right)  ,\cdots,\left(  N-1,N\right)
,\left(  N,1\right)  $, with $N\geq3$. By definition, $J_{1,N}>0$, and
$J_{n,n+1}>0$ for all $1\leq n\leq N-1$. A non-zero $H$ can then be
constructed as%
\[
H_{i,j}=\left\{
\begin{array}
[c]{rl}%
\varepsilon\text{,} & \text{ if }1\leq i\leq N-1\text{ and }j=i+1\text{, }\\
-\varepsilon\text{,} & \text{ if }2\leq i\leq N\text{ and }j=i-1\text{, }\\
\varepsilon\text{,} & \text{ if }i=N\text{ and }j=1\text{,}\\
-\varepsilon\text{,} & \text{ if }i=1\text{ and }j=N\text{,}\\
0\text{,} & \text{ otherwise.}%
\end{array}
\right.
\]
Here
\[
\varepsilon=\min_{1\leq n\leq N-1}\left\{  J_{n,n+1}\text{, }1-J_{n,n+1}%
\text{, }J_{1,N}\text{, }1-J_{1,N}\right\}  \text{.}%
\]
Clearly, $\varepsilon>0$, since all the items in the minimum are above $0$. It
is trivial to verify that $H=-H^{\top}$ and $H\mathbf{1}=0$.

Necessity: Assume there are no loops in $\mathcal{G}_{T}$, then all states in
the chain must be organized in a tree, following the ergodic assumption. In
other word, there are exactly $2\left(  S-1\right)  $ non-zero off-diagonal
elements in $J$. Plus, these $2\left(  S-1\right)  $ elements are arranged
symmetrically along the diagonal and spanning every column and row of $J$.

Because the states are organized in a tree, there is at least one leaf node
$s$ in $\mathcal{G}_{T}$, with a single neighbor $s^{\prime}$. Row $s$ and
column $s$ in $J$ thus looks like $r_{s}=\left[  \cdots,p_{s,s},\cdots
,p_{s,s^{\prime}},\cdots\right]  $ and its transpose, respectively, with
$p_{s,s}\geq0$ and $p_{s,s^{\prime}}>0$, and all other entries being $0$.

Assume that one wants to construct a some $H$, such that $J\pm H\geq0$. Let
$h_{s}$ be the $s$-th row of $H$. Since $r_{s}\pm h_{s}\geq0$, all except the
$s^{\prime}$-th elements in $h_{s}$ must be $0$. But since $h_{s}\mathbf{1}%
=0$, the whole $s$-th row, thus the $s$-th column of $H$ must be $0$.

Having set the $s$-th column and row of $H$ to $0$, one can consider the
reduced Markov chain with one state less, and repeat with another leaf node.
Working progressively along the tree, it follows that all rows and columns in
$H$ must be $0$.
\end{proof}

The indication of Proposition \ref{prop-3} together with \ref{prop-2} is that
all reversible chains can be improved in terms of asymptotic variance using
Corollary \ref{cor-1}, except those whose transition graphs are trees. In
practice, the non-tree constraint is not a problem because almost all current
methods of constructing reversible chains generate chains with loops.

\subsection{Graphical interpretation\label{graphical}}

In this subsection we provide a graphical interpretation of the results in the
previous sections. Starting from a simple case, consider a reversible Markov
chain with three states forming a loop. Let $u_{1}=\left[  1,0,-1\right]
^{\top}$ and $u_{2}=\left[  0,1,-1\right]  ^{\top}$. Clearly, $u_{1}$ and
$u_{2}$ are linearly independent and $u_{1}^{\top}\mathbf{1}=u_{2}^{\top
}\mathbf{1}=0$. By Proposition \ref{prop-2} and \ref{prop-3}, there exists
some $\varepsilon>0$, such that $H=\varepsilon U_{12}$ satisfies Condition I
and II, with $U_{1,2}=u_{1}u_{2}^{\top}-u_{2}u_{1}^{\top}$. Write $U_{1,2}$
and $J+H$ in explicit form,
\[
U_{1,2}=\left[
\begin{array}
[c]{rrr}%
0 & 1 & -1\\
-1 & 0 & 1\\
1 & -1 & 0
\end{array}
\right]  \text{, \ \ }J+H=\left[
\begin{array}
[c]{ccc}%
p_{1,1} & p_{1,2}+\varepsilon & p_{1,3}-\varepsilon\\
p_{2,1}-\varepsilon & p_{2,2} & p_{2,3}+\varepsilon\\
p_{3,1}+\varepsilon & p_{3,2}-\varepsilon & p_{3,3}%
\end{array}
\right]  \text{,}%
\]
with $p_{i,j}$ being the probability of the consecutive states being $i$, $j$.
It is clear that in $J+H$, the probability of jumps $1\rightarrow2$,
$2\rightarrow3$, and $3\rightarrow1$ is increased, and the probability of
jumps in the opposite direction is decreased. Intuitively, this amounts to
adding a `vortex' of direction $1\rightarrow2\rightarrow3\rightarrow1$ in the
state transition. Similarly, the joint probability matrix for the reverse
operator is $J-H$, which adds a vortex in the opposite direction. This simple
case also gives an explanation of why adding or subtracting non-zero $H$ can
only be done where a loop already exists, since the operation requires
subtracting $\varepsilon$ from all entries in $J$ corresponding to edges in
the loop.

In the general case, define $S-1$ vectors $u_{1},\cdots,u_{S-1}$ as%
\[
u_{s}=[0,\cdots,0,\underset{s\text{-th element}}{1},0,\cdots,0,-1]^{\top
}\text{.}%
\]
It is straightforward to see that $u_{1},\cdots,u_{S-1}$ are linearly
independent and $u_{s}^{\top}\mathbf{1}=0$ for all $s$, thus any $H$
satisfying Condition I can be represented as the linear combination of
$U_{i,j}=u_{i}u_{j}^{\top}-u_{j}u_{i}^{\top}$, with each $U_{i,j}$ containing
$1$'s at positions $\left(  i,j\right)  $, $\left(  j,S\right)  $, $\left(
S,i\right)  $, and $-1$'s at positions $\left(  i,S\right)  $, $\left(
S,j\right)  $, $\left(  j,i\right)  $. It is easy to verify that adding
$\varepsilon U_{i,j}$ to $J$ amounts to introducing a vortex of direction
$i\rightarrow j\rightarrow S\rightarrow i$, and any vortex of $N$ states
($N\geq3$) $s_{1}\rightarrow s_{2}\rightarrow\cdots\rightarrow s_{N}%
\rightarrow s_{1}$ can be represented by the linear combination $\sum
_{n=1}^{N-1}U_{s_{n},s_{n+1}}$ in the case of state $S$ being in the vortex
and assuming $s_{N}=S$ without loss of generality, or $U_{s_{N},s_{1}}%
+\sum_{n=1}^{N-1}U_{s_{n},s_{n+1}}$ if $S$ is not in the vortex, as
demonstrated in Figure~\ref{Fig.1}. Therefore, adding or subtracting an $H$ to
$J$ is equivalent to inserting a number of vortices into the state transition map.

\begin{figure}[t]
\begin{center}
\includegraphics[width=5.5in]{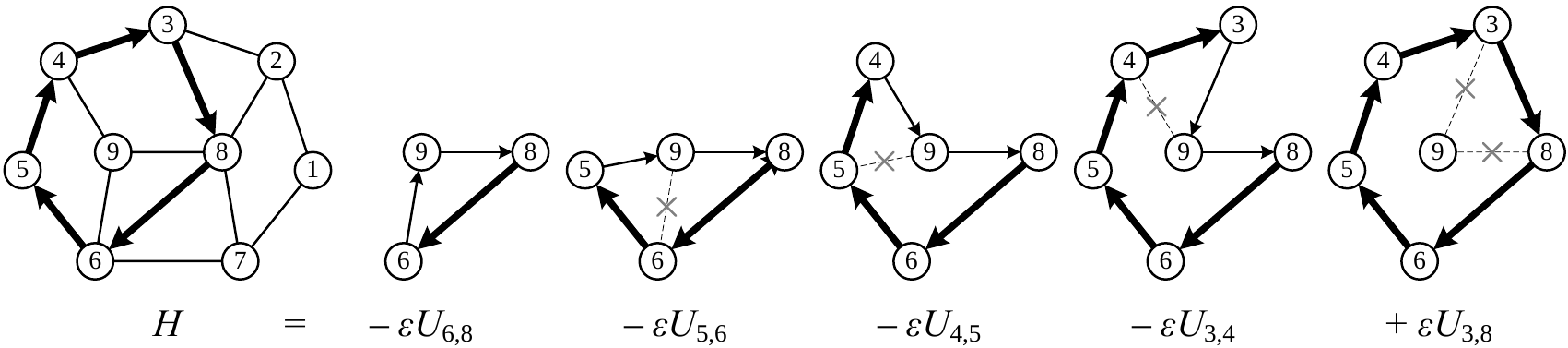}
\end{center}
\caption{Illustration of the construction of larger vortices. The left hand
side is a state transition graph of a reversible Markov chain with $S=9$
states, with a vortex $3\rightarrow8\rightarrow6\rightarrow5\rightarrow4$ of
strength $\varepsilon$ inserted. The corresponding $H$ can be expressed as the
linear combination of $U_{i,j}$, as shown on the right hand side of the graph.
We start from the vortex $8\rightarrow6\rightarrow9\rightarrow8$, and add one
vortex a time. The dotted lines correspond to edges on which the flows cancel
out when a new vortex is added. For example, when vortex $6\rightarrow
5\rightarrow9\rightarrow6$ is added, edge $9\rightarrow6$ cancels edge
$6\rightarrow9$ in the previous vortex, resulting in a larger vortex with four
states. Note that in this way one can construct vortices which do not include
state $9$, although each $U_{i,j}$ is a vortex involving $9$.}%
\label{Fig.1}%
\end{figure}

\section{An example}

Adding vortices to the state transition graph forces the Markov chain to move
in loops following pre-specified directions. The benefit of this can be
illustrated in the following example. Consider a reversible Markov chain with
$S$ states forming a ring, namely from state $s$ one can only jump to
$s\oplus1$ or $s\ominus1$, with $\oplus$ and $\ominus$ being the mod-$S$
summation and subtraction. The only possible non-zero $H$ in this example is
of form $\varepsilon\sum_{s=1}^{S-1}U_{s,s+1}$, corresponding to vortices on
the large ring.

We assume uniform stationary distribution $\pi\left(  s\right)  =\frac{1}{S}$.
In this case, any reversible chain behaves like a random walk. The chain which
achieves minimal asymptotic variance is the one with the probability of both
jumping forward and backward being $\frac{1}{2}$. The expected number of steps
for this chain to reach the state $\frac{S}{2}$ edges away is $\frac{S^{2}}%
{4}$. However, adding the vortex reduces this number to roughly $\frac
{S}{2\varepsilon}$ for large $S$, suggesting that it is much easier for the
non-reversible chain to reach faraway states, especially for large $S$. In the
extreme case, when $\varepsilon=\frac{1}{2}$, the chain cycles
deterministically, reducing asymptotic variance to zero. Also note that the
reversible chain here has zero probability of staying in the current state,
thus cannot be further improved using Peskun's theorem.

Our intuition about why adding vortices helps is that chains with vortices
move faster than the reversible ones, making the function values of the
trajectories less correlated. This effect is demonstrated in
Figure~\ref{Fig.2}.

\begin{figure}[t]
\begin{center}
\includegraphics[width=5.5in]{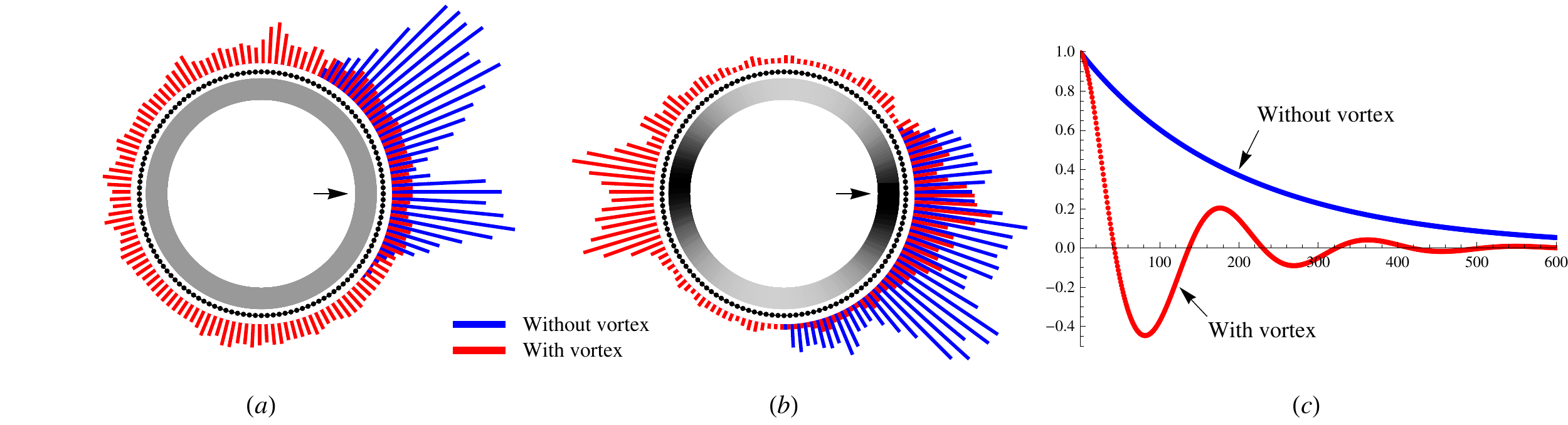}
\end{center}
\caption{Demonstration of the vortex effect: (a) and (b) show two different,
reversible Markov chains, each containing $128$ states connected in a ring.
The equilibrium distribution of the chains is depicted by the gray inner
circles; darker shades correspond to higher probability. The equilibrium
distribution of chain (a) is uniform, while that of (b) contains two peaks
half a ring apart. In addition, the chains are constructed such that the
probability of staying in the same state is zero. In each case, two
trajectories, of length $1000$, are generated from the chain with and without
the vortex, starting from the state pointed to by the arrow. The length of the
bar radiating out from a given state represents the relative frequency of
visits to that state, with red and blue bars corresponding to chains with and
without vortex, respectively. It is clear from the graph that trajectories
sampled from reversible chains spread much slower, with only $1/5$ of the
states reached in (a) and $1/3$ in (b), and the trajectory in (b) does not
escape from the current peak. On the other hand, with vortices added,
trajectories of the same length spread over all the states, and effectively
explore both peaks of the stationary distribution in (b). The plot (c) show
the correlation of function values (normalized by variance) between two states
$\tau$ time steps apart, with $\tau$ ranging from $1$ to $600$. Here we take
the Markov chains from (b) and use function $f\left(  s\right)  =\cos\left(
4\pi\cdot\frac{s}{128}\right)  $. When vortices are added, not only do the
absolute values of the correlations go down significantly, but also their
signs alternate, indicating that these correlations tend to cancel out in the
sum of Eq.\ref{eq-50}.}%
\label{Fig.2}%
\end{figure}

\section{Conclusion}



In this paper, we have presented a new way of converting a reversible finite
Markov chain into a non-reversible one, with the theoretical guarantee that
the asymptotic variance of the MCMC estimator based on the non-reversible
chain is reduced. The method is applicable to any reversible chain whose
states are not connected through a tree, and can be interpreted graphically as
inserting vortices into the state transition graph.

The results confirm that non-reversible chains are fundamentally better than
reversible ones. The general framework of Proposition \ref{prop-1} suggests
further improvements of MCMC's asymptotic performance, by applying other
results from matrix analysis to asymptotic variance reduction. The combined
results of Corollary \ref{cor-1}, and Propositions \ref{prop-2} and
\ref{prop-3}, provide a specific way of doing so, and pose interesting
research questions. Which combinations of vortices yield optimal improvements
for a given chain? Finding one of them is a combinatorial optimization
problem. How can a good combination be constructed in practice, using limited
history and computational resources?

\newpage

\subsubsection*{References}

\renewcommand*{\refname}{\vspace*{-10mm}}

\end{document}